\documentclass[12pt, final]{l4dc2022}
 
\usepackage{latexsym}
\usepackage{graphicx}
\usepackage{color}
\usepackage{url}
\usepackage{Shorthands}
\usepackage{breqn}
\usepackage{caption}
\usepackage{makecell}

\usepackage{listings}
\usepackage{courier}
\usepackage{comment}
 \usepackage{dblfnote}
\setcitestyle{square}
\usepackage[ruled]{algorithm2e}
\SetKwInput{KwInput}{Input}                
\SetKwInput{KwOutput}{Output}              
\usepackage{hyperref}
\hypersetup{%
    pdfborder = {0 0 0}
}

\renewcommand{\hat}{\widehat}
\renewcommand{\tilde}{\widetilde}

\usepackage{mathtools}

\newcommand{\qed}{\hfill\blacksquare}

\newtheorem{mytheorem}{Theorem}
\newtheorem{mylemma}{Lemma}
\newtheorem{myremark}{Remark}
\newtheorem{assumption}{Assumption}

     \makeatletter
    \def\footnoterule{\kern-1\p@
      \hrule \@width 2in \kern 2.6\p@} 
    \makeatother

\title[Adaptive Stochastic MPC under Unknown Noise Distribution]{Adaptive Stochastic MPC under Unknown Noise Distribution}
\usepackage{times}

\author{%
 \Name{Charis Stamouli} \Email{stamouli@seas.upenn.edu}\\
 \Name{Anastasios Tsiamis} \Email{atsiamis@seas.upenn.edu}\\
 \Name{Manfred Morari} \Email{morari@seas.upenn.edu}\\
 \Name{George J. Pappas} \Email{pappasg@seas.upenn.edu}\\
 \addr University of Pennsylvania, Philadelphia, PA
}

\begin{document}

\maketitle

\begin{abstract}%
In this paper, we address the stochastic MPC (SMPC) problem for linear systems, subject to chance state constraints and hard input constraints, under unknown noise distribution. First, we reformulate the chance state constraints as deterministic constraints depending only on explicit noise statistics. Based on these reformulated constraints, we design a distributionally robust and robustly stable benchmark SMPC algorithm for the ideal setting of known noise statistics. Then, we employ this benchmark controller to derive a novel robustly stable adaptive SMPC scheme that learns the necessary noise statistics online, while guaranteeing time-uniform satisfaction of the unknown reformulated state constraints with high probability. The latter is achieved through the use of confidence intervals which rely on the empirical noise statistics and are valid uniformly over time. Moreover, control performance is improved over time as more noise samples are gathered and better estimates of the noise statistics are obtained, given the online adaptation of the estimated reformulated constraints. Additionally, in tracking problems with multiple successive targets our approach leads to an online-enlarged domain of attraction compared to robust tube-based MPC. A numerical simulation of a DC-DC converter is used to demonstrate the effectiveness of the developed methodology.
\end{abstract}


\captionsetup[figure]{labelfont={bf},labelformat={default},labelsep=colon,name={Fig.}}
\captionsetup[table]{labelfont={bf},labelformat={default},labelsep=colon,name={Table}}
\section{Introduction}
Model predictive control (MPC) has significantly matured in recent decades, establishing itself as the primary methodology for optimal control of multivariable constrained systems. However, model uncertainty poses a challenge to system performance under classical MPC frameworks, where a model description derived offline is used during control implementation. To mitigate the effects of uncertainties and improve system performance, various adaptive MPC strategies, which update the system model online, have been developed.

The majority of works on adaptive MPC focus on robust MPC, where hard constraints are enforced \citep{Fukushima2007,Weiss2014,Morari2014Adaptive,Desaraju2017}. While well established theoretically, robust MPC can be excessively conservative in practice, given its worst-case treatment of the noise which often accounts for unlikely uncertainty realizations. Stochastic MPC (SMPC) reduces this conservatism and increases control authority by allowing constraint violations via the use of chance constraints. Examples of applications where SMPC has been successfully used in the past include building climate control \citep{Borelli2012}, autonomous vehicle control \citep{Borelli2014}, and wind turbine control \citep{Kou2016}.

In the realm of stochastic model predictive control, several approaches relying on the reformulation of chance constraints into deterministic counterparts have been proposed in both the classical \citep{Korda2011,Paulson2019} and the adaptive \citep{Borelli2021,Spanos2018} regime. The main limitation of these algorithms is the restriction to special noise distributions, since analytic constraint reformulations involve the complex computation of multivariate integrals. The authors in \citep{Kouvaritakis2010} overcome this challenge through the use of numerical approximations based on the noise distribution. However, since the distribution is itself subject to uncertainty, assuming perfect knowledge of it is unrealistic, and in fact, even small errors can dramatically change the control inputs selected by a model predictive controller \citep{VanParys2016}. On the other hand, if the noise distribution is assumed to be unknown, the exact expression of the chance constraints is also unknown.

Recent work in \citep{Santos2020} addresses the above issue by replacing the true noise distribution with the empirical one and using confidence intervals to bound the uncertainty in the deterministically-reformulated chance constraints. The algorithm learns the reformulated constraints online and ensures satisfaction of the unknown chance constraints at each time step with high probability. Notice though that given such per-time-step high-confidence guarantees, the probability of time-uniform satisfaction of the chance constraints decays with system operation time.

An alternative to chance constraints reformulation that can be employed under unknown noise distribution is provided by scenario-based SMPC \citep{Calafiore2006,Morari2012,Morari2014Scenario,Lorenzen2017}. In scenario SMPC formulations, hard constraints are imposed for a finite number of noise samples (scenarios), which are drawn either empirically or through a random number generator. Nevertheless, unlike in \citep{Santos2020}, high-confidence satisfaction of the unknown chance constraints at each time step is ensured only for a sufficiently large number of noise samples.

In this paper, we address the problem of stochastic MPC for linear systems, subject to chance state constraints and hard input constraints, under unknown noise distribution. Our goal is to learn and adapt to the unknown noise online so that control performance is gradually improved, while also maintaining safety at all time steps with high probability. Our contributions are the following:
\vspace*{-0.2cm}
\begin{enumerate}[label=\roman*)]
    \item For the ideal setting of known noise statistics we propose a distributionally robust benchmark SMPC scheme which also guarantees recursive feasibility and robust stability. This is achieved by replacing the chance state constraints with novel distribution-free deterministic constraints relying only on explicit noise statistics, in contrast to \citep{Kouvaritakis2010,Lorenzen2017b}, where distribution-dependent constraints are employed instead. 
    \vspace{-0.3cm}
    \item For the actual setting of unknown noise statistics we introduce a novel adaptive SMPC algorithm that estimates the necessary noise statistics and learns the reformulated constraints using noise samples collected online. Recursive feasibility and robust stability are ensured, while control performance is improved over time through adaptation. Additionally, in tracking problems with multiple successive targets our scheme leads to an online-enlarged domain of attraction compared to robust tube-based MPC.
    \vspace{-0.3cm}
    \item By exploiting time-uniform confidence intervals based on the empirical noise statistics, we bound the uncertainty in the reformulated constraints uniformly over time with high probability. Note that in contrast to \citep{Santos2020}, this probability is constant and does not decrease with system operation time. To the best of our knowledge, this is the first work in the SMPC literature that guarantees satisfaction of unknown chance state constraints at all times steps jointly with high confidence. 
\end{enumerate}
\vspace{-0.15cm}
We exhibit the  benefits of the proposed scheme through numerical simulations. All proofs can be found in Appendix A.

\section{Problem Formulation}
Consider a discrete-time linear time-invariant system with additive noise, described by the following state-space model:
\begin{equation}\label{system}
    x_{t+1} = Ax_t+Bu_t+w_t,
\end{equation}
where $x_t\in\setR^n$ is the state, $u_t\in\setR^m$ is the control input, and $A,\;B$ are system matrices of appropriate dimensions. The initial state $x_0$ is assumed to be deterministic and fixed. Moreover, we assume that the process noise $w_t\in\setR^n$ is i.i.d. with unknown probability distribution supported on the hypercube $\calW = \{w\in\setR^n \where \norm{w}_\infty\leq a_w\}$, where $a_w$ is a known positive noise bound.

\begin{assumption} The noise distribution is unknown. The system matrices $A,\;B$ and the noise bound $a_w$ are known. The noise mean is also known and equal to zero, that is, $\mu_w:=\Exp w_t=0$.
\end{assumption}
\begin{myremark} An extension for unknown noise mean $\mu_w$ is presented in Appendix B. Note that in practice the noise bound $a_w$ could be a rough estimate of the noise components' distribution support that we know from system structure or have obtained from observations. Our results can easily be extended to the case of general polytopes $\calW$ that contain the origin. 
\end{myremark}

The state is assumed to be subject to chance constraints of the form:
\begin{equation}\label{chance_state_constraints}
    \Prob\left(g_j^Tx_t\leq h_j\right)\geq1-\alpha_j,
\end{equation}
for $j=1,\ldots,n_x$, whereas the control input should satisfy hard constraints of the form:
\begin{equation}\label{input_constraints}
    u_t\in\calU,
\end{equation}
for all $t\geq0$, where $\alpha_j\in(0,1)$ is the maximum admissible probability of violating the constraint $g_j^Tx_t\leq h_j$, and $\calU\subseteq\setR^m$ is a polytope containing the origin. The probabilities $\alpha_j$ are specified directly from application requirements. Note that given the assumption that the noise distribution is unknown, the exact expression of the chance state constraints \eqref{chance_state_constraints} is also unknown. Hence, herein we refer to \eqref{chance_state_constraints} as \textit{unknown} chance state constraints. Meanwhile, we consider the control performance objective determined by the cost function:
\begin{equation*}
    \ell(x_t,u_t) = x_t^TQx_t+u_t^TRu_t,
\end{equation*}
where $Q\in\setR^{n\times n}$ and $R\in\setR^{m\times m}$ are positive-definite matrices. Notice that using the state measurements we can sequentially and online collect the noise samples $w_0,\ldots,w_{t-1}$, as we control the system up to time step $t$.

The goal of this work is to develop, under the above assumptions, a recursively feasible and robustly stable adaptive SMPC algorithm for system \eqref{system}. In particular, the algorithm should respect the input constraints \eqref{input_constraints} as well as certain state constraints which are designed online so that:
\begin{enumerate}[label=\roman*)]
    \vspace{-0.2cm}
    \item control performance is improved over time, as more noise samples are collected and better estimates of explicit noise statistics are obtained;
    \vspace{-0.3cm}
    \item time-uniform satisfaction of \eqref{chance_state_constraints} (i.e., satisfaction of \eqref{chance_state_constraints} at all time steps) is guaranteed with high confidence.
\end{enumerate}
\vspace{-0.2cm}
The benchmark SMPC scheme introduced in Section 4 for the ideal case of known noise statistics is used as reference for evaluating the control performance of our adaptive SMPC scheme (see Section 5) for unknown noise statistics.

\section{Robust Tube-based MPC}
In this section, we briefly review the classical robust tube-based MPC \citep{Mayne2005}, which we leverage in the design of our benchmark SMPC algorithm to provide recursive feasibility and robust stability guarantees (see Section 4 for details). The main difference between our framework and robust tube-based MPC lies in the formulation of the state constraints. Particularly, the authors in \citep{Mayne2005} consider hard state constraints of the form:
\begin{equation}\label{hard_state_constraints}
    x_t\in\calX = \left\{x\in\setR^n \where g_j^Tx\leq h_j,\;j=1,\ldots,n_x\right\}
\end{equation}
instead of chance state constraints of the form \eqref{chance_state_constraints}. In addition, they only assume that the noise distribution is bounded, thus adopting a worst-case noise model, in contrast to our approach where a stochastic noise model is employed instead. 

Let $\bar{x}_t$ denote the nominal (noise-free) state, $\bar{u}_t=u_t\left(\bar{x}_t\right)$ the corresponding nominal input, and $e_t=x_t-\bar{x}_t$ the difference between the actual and the nominal state. We select gains $K,\;\bar{K}\in\setR^{m\times n}$ such that the system matrices $A_{cl}:=A+BK,\;\bar{A}_{cl} := A+B\bar{K}$ are stable. For any vector $z$, we use $z_{t+k|t}$ to denote the $k$-step ahead prediction of $z$ made at time step $t$ by some model predictive controller (i.e., the prediction of $z_{t+k}$ made at time step $t$). At each time step $t$ the robust tube-based MPC solves the following finite horizon optimal control problem:
\begin{align}\label{RMPC}
    \min_{\substack{\bar{x}_{t|t},\ldots,\bar{x}_{t+N|t} \\ \bar{u}_{t|t},\ldots,\bar{u}_{t+N-1|t}}}\; &\bar{x}_{t+N|t}^TP\bar{x}_{t+N|t}+\sum_{k=0}^{N-1}\ell(\bar{x}_{t+k|t},\bar{u}_{t+k|t}) \nonumber \\
    \subto &\bar{x}_{t+k+1|t} = A\bar{x}_{t+k|t}+B\bar{u}_{t+k|t} \\
        & \bar{x}_{t+k|t}\in\bar{\calX}_R,\;k=1,\ldots,N-1 \nonumber \\
        &\bar{u}_{t+k|t}\in \calU\ominus K\calE,\;k=0,\ldots,N-1 \nonumber \\
        &\bar{x}_{t+N|t}\in\bar{\calX}_{f,R} \nonumber \\
        &x_t\in\bar{x}_{t|t}\oplus\calE, \nonumber
\end{align}
where $\bar{x}_{t+k|t},\;\bar{u}_{t+k|t}$ denote the prediction of $\bar{x}_{t+k},\;\bar{u}_{t+k}$ made at time step $t$, respectively, $\calE$ is the minimum robust invariant set \citep[Definition 3]{Rakovic2004} for the error system:
\begin{equation}\label{error_system}
    e_{t+k+1|t} = A_{cl}e_{t+k|t}+w_{t+k|t},
\end{equation}
and $\bar{\calX}_R=\calX\ominus\calE$. The set $\bar{\calX}_{f,R}$ is the maximum invariant set \citep[Definition 2.2]{Kerrigan2000} for the system $\bar{x}^+=\bar{A}_{cl}\bar{x}$ that satisfies the constraints $\bar{\calX}_{f,R}\subseteq\bar{\calX}_R$ and $\bar{K}\bar{\calX}_{f,R}\subseteq \calU\ominus K\calE$, and $P$ is the unique positive definite solution of the Lyapunov equation:
\begin{equation}\label{lyapunov_equation}
   \bar{A}_{cl}^TP\bar{A}_{cl}-P= -Q-\bar{K}^TR\bar{K}.
\end{equation}
After we obtain the optimal sequences $\{\bar{x}_{t|t}^*,\ldots,\bar{x}_{t+N|t}^*\}$ and $\{\bar{u}_{t|t}^*,\ldots,\bar{u}_{t+N-1|t}^*\}$ of \eqref{RMPC}, the tube-based state-feedback control policy:
\begin{equation}\label{tube_controller}
    u_t\left(x_t\right) = K\left(x_t-\bar{x}_{t|t}^*(x_t)\right)+\bar{u}_{t|t}^*(x_t)
\end{equation}
is applied to system \eqref{system}. Notice in \eqref{tube_controller} that the first term on the right-hand side keeps the actual state close to the nominal state, whereas the second term steers the nominal state to the origin. At the next time step we repeat the same process, thus yielding a receding horizon strategy. Employing \eqref{error_system}, one can write the halfspace inequalities  determining the nominal state constraints of \eqref{RMPC} as follows:
\begin{equation}\label{RMPC_state_constraints}
    g_j^T\bar{x}_{t+k|t}\leq h_j-\gamma_j-\norm{g_j}_1a_w,
\end{equation}
where $\gamma_j = \max\limits_{e\in\calE}g_j^TA_{cl}e$, for all $j=1,\ldots,n_x$. In the next section, stochasticity is employed to relax these constraints, thus reducing the conservatism of robust MPC \citep{Mayne2005}.

\section{Tube-based SMPC under Known Noise Statistics}
In this section, we consider the ideal setting where the noise covariance $\Sigma_w:= \mathbb{E}w_tw^T_t$ is known and introduce a
stochastic model predictive controller with polytopic state constraints that depend only on $\Sigma_w$ and are sufficient for the original state constraints \eqref{chance_state_constraints}. Herein, we refer to this controller as \textit{benchmark SMPC}. This is because the performance of this controller will be used as reference for evaluating the performance of our adaptive SMPC scheme, which is introduced in the next section for the actual setting where $\Sigma_w$ is unknown. More specifically, inspired by the idea of constraint tightening employed in robust tube-based MPC, we derive a less conservative robustly stable MPC algorithm, by imposing suitable mixed probabilistic/worst-case constraint tightening to \eqref{chance_state_constraints} based on $\Sigma_w$. As we show next, our approach leads to a relaxation of the nominal state constraints \eqref{RMPC_state_constraints} of the robust tube-based MPC, where worst-case constraint tightening was instead applied to \eqref{hard_state_constraints}. 

\begin{mylemma}
Given system \eqref{system} and a prediction horizon of length $N$, at each time step $t$ the $k$-step ahead prediction of the actual state $x_{t+k|t}$ satisfies:
\begin{equation*}
    \Prob\left(g_j^Tx_{t+k|t}\leq h_j\right)\geq1-\alpha_j,
\end{equation*}
if $x_t\in\bar{x}_{t|t}\oplus\calE$ and the $k$-step ahead prediction of the nominal state $\bar{x}_{t+k|t}$ satisfies:
\begin{equation}\label{pre_BSMPC_nominal_state_constraints}
    g_j^T\bar{x}_{t+k|t}\leq h_j-\gamma_j-f(\alpha_j)\sqrt{g_j^T\Sigma_wg_j},
\end{equation}
with $\gamma_j$ as in \eqref{RMPC_state_constraints} and $f(\alpha_j)=\sqrt{\frac{1-\alpha_j}{\alpha_j}}$, for all $j=1,\ldots,n_x$ and $k=1,\ldots,N$.
\end{mylemma}

We want to employ the constraints \eqref{pre_BSMPC_nominal_state_constraints} to develop a simple stochastic tube-based MPC algorithm that is less conservative than the robust tube-based MPC. This means that \eqref{pre_BSMPC_nominal_state_constraints} should be looser than \eqref{RMPC_state_constraints}, that is, we would like to have $ f(\alpha_j)\sqrt{g_j^T\Sigma_wg_j}\leq \norm{g_j}_1a_w$, for all $\alpha_j\in(0,1)$. However, it is clear from the form of $f(\cdot)$ that this inequality may not be satisfied for  small values of $\alpha_j$. Hence, in order to prevent excessive conservatism, we replace the constraints \eqref{pre_BSMPC_nominal_state_constraints} with:
\begin{align}\label{BSMPC_nominal_state_constraints}
    g_j^T\bar{x}_{t+k|t}\leq h_j-\gamma_j-\min\left\{\norm{g_j}_1a_w,f(\alpha_j)\sqrt{g_j^T\Sigma_wg_j}\right\}.
\end{align}
\begin{myremark}
Notice in \eqref{BSMPC_nominal_state_constraints} that the first member of the minimum (robust) will be active for small values of $\alpha_j$ and the second one (stochastic) otherwise. This is reasonable, because selecting very small maximum admissible probabilities of constraint violation $\alpha_j$ is approximately equivalent to imposing hard state constraints, which are guaranteed to be satisfied due to the robust member of the minimum.
\end{myremark}
\noindent Let $\bar{\calX}_S$ denote the polytope defined by the halfspace inequalities \eqref{BSMPC_nominal_state_constraints}, and $\bar{\calX}_{f,S}$ the maximum invariant set for the system $\bar{x}^+=\bar{A}_{cl}\bar{x}$ that satisfies the constraints $\bar{\calX}_{f,S}\subseteq\bar{\calX}_S$ and $\bar{K}\bar{\calX}_{f,S}\subseteq \calU\ominus K\calE$. At each time step $t$ our benchmark SMPC solves the following optimal control problem:
\begin{align}\label{BSMPC}
    \min_{\substack{\bar{x}_{t|t},\ldots,\bar{x}_{t+N|t} \\ \bar{u}_{t|t},\ldots,\bar{u}_{t+N-1|t}}}\; &\bar{x}_{t+N|t}^TP\bar{x}_{t+N|t}+\sum_{k=0}^{N-1}\ell(\bar{x}_{t+k|t},\bar{u}_{t+k|t}) \nonumber \\
    \subto &\bar{x}_{t+k+1|t} = A\bar{x}_{t+k|t}+B\bar{u}_{t+k|t} \\
        & \bar{x}_{t+k|t}\in\bar{\calX}_S,\;k=1,\ldots,N-1 \nonumber \\
        &\bar{u}_{t+k|t}\in \calU\ominus K\calE,\;k=0,\ldots,N-1 \nonumber \\
        &\bar{x}_{t+N|t}\in\bar{\calX}_{f,S} \nonumber \\
        &x_t\in\bar{x}_{t|t}\oplus\calE \nonumber
\end{align}
and then the state-feedback control law  \eqref{tube_controller} is applied to system \eqref{system}.
\begin{mytheorem}[Benchmark SMPC]
Suppose the noise covariance $\Sigma_w$ is known and the optimization problem \eqref{BSMPC} is feasible at time step $t=0$. Then, \eqref{BSMPC} is feasible and ensures satisfaction of the original chance state constraints \eqref{chance_state_constraints} at every time step $t\geq0$. Moreover, system \eqref{system} in closed-loop with the model predictive controller defined by \eqref{BSMPC} and \eqref{tube_controller} asymptotically converges to the set $\calE$ for all noise realizations. 
\end{mytheorem}

\begin{myremark}
The problem of stochastic MPC with stability guarantees was addressed in \citep{Kouvaritakis2010}, where distribution-dependent reformulated state constraints were employed. In contrast, the reformulated state constraints \eqref{BSMPC_nominal_state_constraints} of our benchmark SMPC scheme are distribution-free in the sense that they guarantee satisfaction of the original state constraints \eqref{chance_state_constraints} for all noise distributions with zero mean and covariance $\Sigma_w$. Employing these constraints, we were able to derive a distributionally robust SMPC algorithm, which also enjoys stability guarantees.
\end{myremark}
\section{Adaptive Tube-based SMPC under Unknown Noise Statistics}
In the previous section, we introduced a simple tube-based SMPC algorithm assuming that the noise covariance $\Sigma_w$ is known. We are now ready to present an adaptive SMPC algorithm for unknown $\Sigma_w$ that learns the nominal state constraints \eqref{BSMPC_nominal_state_constraints} of the benchmark SMPC online and guarantees their satisfaction at all time steps with high probability. The main idea is to use noise samples to estimate the term $g_j^T\Sigma_wg_j$ of \eqref{BSMPC_nominal_state_constraints} online, and provide explicit confidence intervals for these estimates so that \eqref{BSMPC_nominal_state_constraints} is always satisfied. As more samples are gathered, the confidence intervals become smaller, thus relaxing the constraints and enhancing the control performance over time. 

We estimate $\Sigma_w$ using the standard sample covariance, which is defined as $\hat{\Sigma}_w^t = \frac{1}{t}\sum_{i=0}^{t-1}w_iw_i^T$, at each time step $t\geq1$. The noise samples $w_0,\ldots,w_{t-1}$ can be computed as $ w_i = x_{i+1}-Ax_i-Bu_i$ for all $t\geq1$. To obtain a confidence interval for the unknown term $g_j^T\Sigma_wg_j$ of \eqref{BSMPC_nominal_state_constraints} based on $\hat{\Sigma}_w^t$, we employ the boundedness of the noise, which allows us to derive the following lemma.
\begin{mylemma}
Fix a failure probability $\delta\in(0,1)$. With probability at least $1-\delta$, the sample covariance $\hat{\Sigma}_w^t$ satisfies: 
\begin{equation}\label{confidence_bound}
    g_j^T\Sigma_wg_j\leq g_j^T\hat{\Sigma}_w^tg_j+r_{tj}(\delta),
\end{equation}
where $r_{tj}(\delta) = \norm{g_j}_1^2a_w^2\sqrt{0.5t^{-1}\log((\pi^4t^2n_x^2)/(36\delta))}$, uniformly over all $t\geq1$ and $j=1,\ldots,n_x$.
\end{mylemma}
Notice in \eqref{confidence_bound} that the bound $g_j^T\hat{\Sigma}_w^tg_j+r_{tj}(\delta)$, where $\delta$ is a design parameter determining its confidence, converges to $g_j^T\Sigma_wg_j$ as fast as $\calO\left(\sqrt{t^{-1}\log t}\right)$. From Lemma 2 we deduce that a sufficient condition for the constraints \eqref{BSMPC_nominal_state_constraints} of the benchmark SMPC to be satisfied for all $t\geq1$ with probability at least $1-\delta$ is the following:
\begin{align}\label{pre_ASMPC_nominal_state_constraints}
    &g_j^T\bar{x}_{t+k|t}\leq h_j-\gamma_j-\min\bigg\{\norm{g_j}_1a_w, f(\alpha_j)\sqrt{g_j^T\hat{\Sigma}_w^tg_j+r_{tj}(\delta)}\bigg\}.
\end{align}
Given that it suffices to update the nominal state constraint set only when the constraints \eqref{pre_ASMPC_nominal_state_constraints} become looser, we define $\bar{\calX}_S^t$ as the intersection of the halfspaces determined by the inequalities:
\begin{align}\label{ASMPC_nominal_state_constraints}
 &g_j^T\bar{x}_{t+k|t}\leq h_j-\gamma_j- \min\left\{\norm{g_j}_1a_w,\min_{1\leq\tau\leq t}f(\alpha_j)\sqrt{g_j^T\hat{\Sigma}_w^\tau g_j+r_{\tau j}(\delta)}\right\},
\end{align}
for each $t\geq1$. Moreover, we define $\bar{\calX}_{f,S}^t$ as the maximum invariant set for the system $\bar{x}^+=\bar{A}_{cl}\bar{x}$ that satisfies the constraints $\bar{\calX}_{f,S}^t\subseteq\bar{\calX}_S^t$ and $\bar{K}\bar{\calX}_{f,S}^t\subseteq \calU\ominus K\calE$, for each $t\geq1$. Since the first noise sample is computed at time step $t=1$, we set $\bar{\calX}_S^0=\bar{\calX}_R$ and $\bar{\calX}_{f,S}^0 = \bar{\calX}_{f,R}$. At each time step $t$ our adaptive SMPC algorithm solves the following finite horizon optimal control problem:
\begin{align}\label{ASMPC}
    \min_{\substack{\bar{x}_{t|t},\ldots,\bar{x}_{t+N|t} \\ \bar{u}_{t|t},\ldots,\bar{u}_{t+N-1|t}}}\; &\bar{x}_{t+N|t}^TP\bar{x}_{t+N|t}+\sum_{k=0}^{N-1}\ell(\bar{x}_{t+k|t},\bar{u}_{t+k|t}) \nonumber \\
    \subto &\bar{x}_{t+k+1|t} = A\bar{x}_{t+k|t}+B\bar{u}_{t+k|t} \\
        & \bar{x}_{t+k|t}\in\bar{\calX}_S^t,\;k=1,\ldots,N-1 \nonumber \\
        &\bar{u}_{t+k|t}\in \calU\ominus K\calE,\;k=0,\ldots,N-1 \nonumber \\
        &\bar{x}_{t+N|t}\in\bar{\calX}_{f,S}^t \nonumber \\
        &x_t\in\bar{x}_{t|t}\oplus\calE \nonumber
\end{align}
and then the state-feedback controller \eqref{tube_controller} is applied to system \eqref{system}.

\begin{mytheorem}[Adaptive SMPC]
Fix a failure probability $\delta\in(0,1)$ and suppose the optimization problem \eqref{ASMPC} is feasible at time step $t=0$. Then, \eqref{ASMPC} is feasible at every time step $t\geq0$, and ensures satisfaction of the unknown original chance state constraints \eqref{chance_state_constraints} uniformly over all $t\geq0$ with probability at least $1-\delta$. Moreover, system \eqref{system} in closed-loop with the model predictive controller defined by \eqref{ASMPC} and \eqref{tube_controller} asymptotically converges to the set $\calE$ for all noise realizations.
\end{mytheorem}

\begin{wrapfigure}{r}{5cm}
   \centering
   \vspace*{-0.4cm}
   \includegraphics[width=0.3\textwidth]{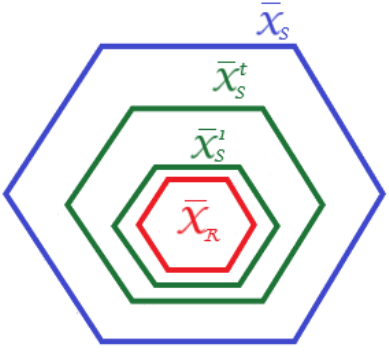}
   \caption{{\small High-confidence sketch illustration of the subset relationship between $\bar{\calX}_R$, $\bar{\calX}_S$, $\bar{\calX}^1_S$, and $\bar{\calX}^t_S$, for some $t\geq2$. The set $\bar{\calX}_S^t$ converges to $\bar{\calX}_S$ as $t\to\infty$.}}
\end{wrapfigure}
\begin{myremark} Notice that the difference between the robust MPC, the benchmark SMPC and the adaptive SMPC lies in the nominal state and terminal constraints. We know that $\hat{\Sigma}_w^t$ is an asymptotically consistent estimate of $\Sigma_w$, and we have $r_{tj}(\delta)=\calO\left(\sqrt{t^{-1}\log t}\right)\xrightarrow{t\to\infty}0$, for all $j=1,\ldots,n_x$. Hence, the nominal state constraints \eqref{ASMPC_nominal_state_constraints} of \eqref{ASMPC} converge to the nominal state constraints \eqref{BSMPC_nominal_state_constraints} of \eqref{BSMPC} as $t\rightarrow\infty$. Moreover, from our theoretical analysis we know that $\bar{\calX}_R=\bar{\calX}_S^0\subseteq\ldots\subseteq\bar{\calX}_S^t$, for all $t\geq1$, and $\bigcup_{t=0}^\infty\bar{\calX}_S^t\subseteq\bar{\calX}_S$ with high probability. The above observations are illustrated in Figure 1. We conclude that both the benchmark and the adaptive SMPC are less conservative than the robust MPC, and in addition, the adaptive SMPC ensures satisfaction of the benchmark SMPC's nominal state constraints at all time steps with high confidence, while learning them online. Note that, herein, conservatism is related to the size of the feasible set of each MPC problem in the sense that the larger the feasible set, the smaller the cost and thus the better the performance of a controller. For this reason, the performance of the adaptive SMPC is expected to be worse than the one of the benchmark SMPC and converge to it as $t\to\infty$. 
\end{myremark}
\begin{myremark}[Tracking of multiple successive targets]
Up to now, we only studied the problem of tracking the origin $x=0$. However, our adaptive SMPC scheme can easily be extended to track non-zero targets. This can be particularly useful in a tracking scenario with multiple successive targets. More specifically, notice that every time we rerun our adaptive SMPC to track a new target, the feasible set of problem \eqref{ASMPC} is larger and larger than the one of the robust MPC problem \eqref{RMPC}, which implies an online enlargement of the domain of attraction. In Section 6, we validate this observation through an illustrative example.  
\end{myremark}

\section{Case Study: DC-DC Converter}
In this section, we illustrate the efficacy of the proposed adaptive SMPC strategy using a simplified model of a DC-DC converter \citep{Lorenzen2015}. Particularly, we consider a linearized DC-DC converter model of the form \eqref{system} with:
\begin{equation*}
    A = \begin{bmatrix}
            1 & 0.0075 \\
            -0.143 & 0.996
        \end{bmatrix},\;
    B = \vectt{4.798}{0.115},
\end{equation*}
and zero-mean Gaussian noise with covariance $\Sigma=0.25\id$ truncated at $a_w=1$. Suppose the state is subject to:
\begin{align*}
    &\Prob\left(\vect{1}{0}x_t\leq9.6\right)\geq0.2,\;\Prob\left(\vect{-1}{0}x_t\leq9.6\right)\geq0.2, \\
    &\Prob\left(\vect{0}{1}x_t\leq9.6\right)\geq0.2,\;\Prob\left(\vect{0}{-1}x_t\leq9.6\right)\geq0.2,
\end{align*}
and the input  to hard constraints given by $\abs{u_t}\leq5$. We wish to minimize the cost for $Q=\diag2(1,10)$ and $R=1$ over a prediction horizon of length $N=9$. The gain matrices were selected as $K=\vect{-0.3429}{0.8629}$ and $\bar{K}=\vect{-0.2858}{0.4910}$, and the failure probability of time-uniform satisfaction of \eqref{BSMPC_nominal_state_constraints} for the adaptive SMPC as $\delta=0.1$.

Figure 2 shows the nominal state and terminal constraint sets of the robust MPC, the benchmark SMPC, and the adaptive SMPC for $10$, $100$ and $1000$ noise samples. Notice that $\bar{\calX}_R\subseteq\bar{\calX}_S^{10}\subseteq\bar{\calX}_S^{100}\subseteq\bar{\calX}_S^{1000}\subseteq\bar{\calX}_S$ and  $\bar{\calX}_{f,R}\subseteq\bar{\calX}_{f,S}^{10}\subseteq\bar{\calX}_{f,S}^{100}\subseteq\bar{\calX}_{f,S}^{1000}\subseteq\bar{\calX}_{f,S}$, which indicates that the robust MPC is more conservative than the adaptive SMPC, which is more conservative than the benchmark SMPC. Hence, the time-uniform high-confidence guarantees of our adaptive SMPC satisfying the unknown chance state constraints \eqref{chance_state_constraints} are validated. Moreover, we observe that the sets $\bar{\calX}_S^{1000}$ and $\bar{\calX}_{f,S}^{1000}$ of the adaptive SMPC are very close to the sets $\bar{\calX}_S$ and $\bar{\calX}_{f,S}$, respectively, of the benchmark SMPC, in accordance with our theoretical result that the nominal state and terminal constraint sets of the adaptive SMPC asymptotically converge to the corresponding sets of the benchmark SMPC.

\begin{figure}[tbh]
\centering
\begin{minipage}{.4\linewidth}
  \centering
  \includegraphics[width=\linewidth]{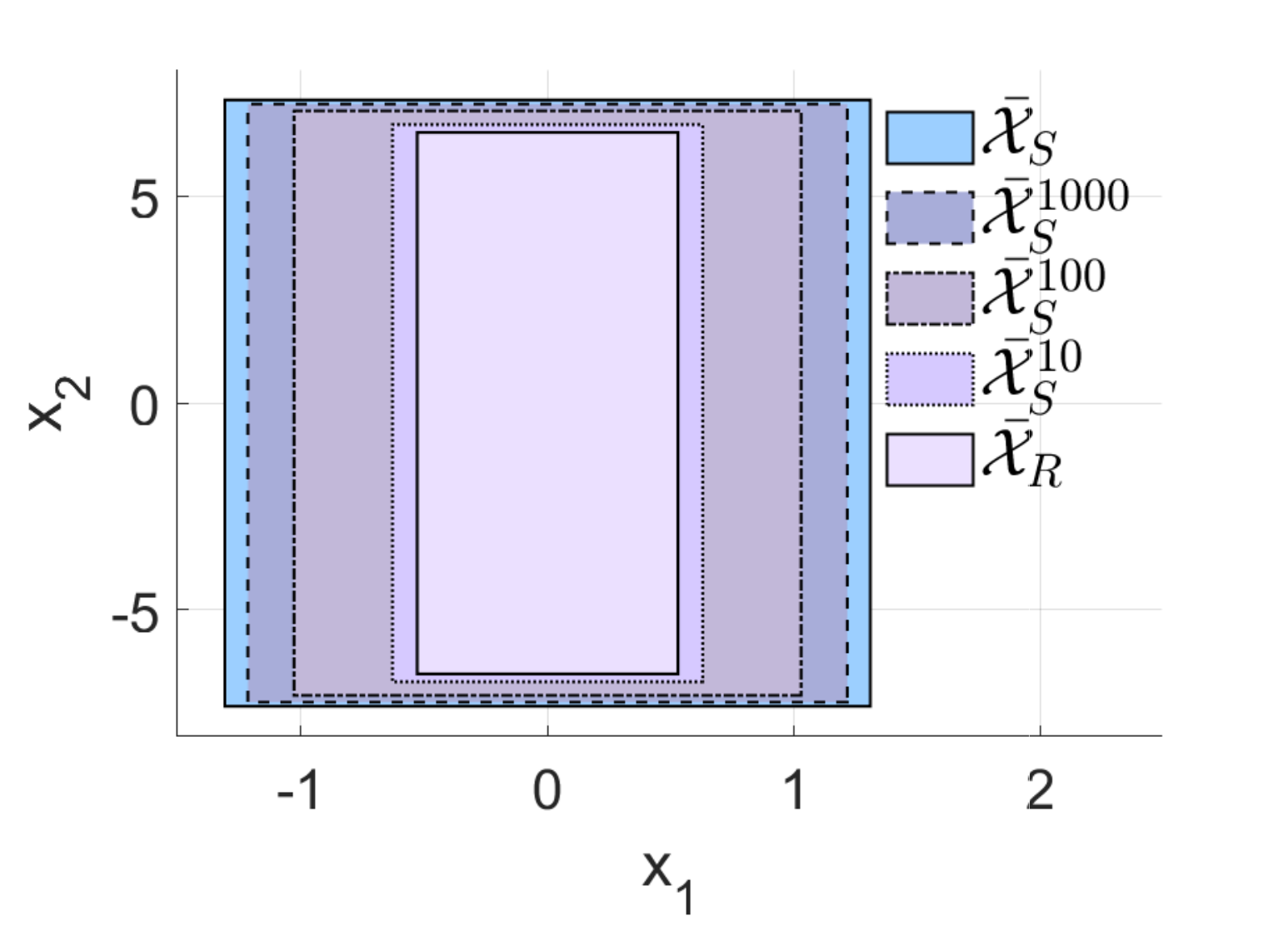}
  \\{\small (a)}
  \label{fig:test1}
\end{minipage}
\hspace{-0cm}
\begin{minipage}{.4\linewidth}
  \centering
  \includegraphics[width=\linewidth]{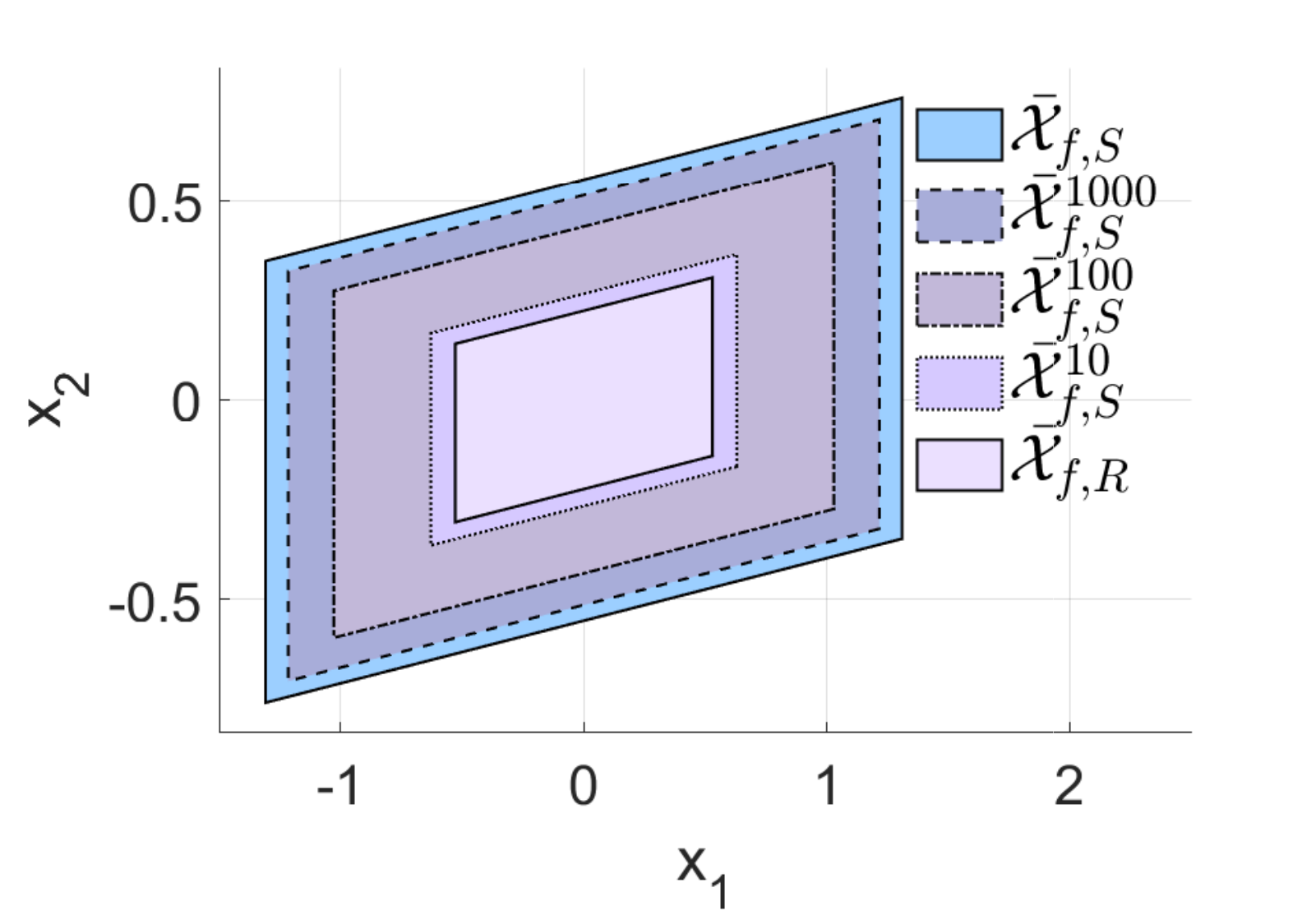}
  \\{\small (b)}
  \label{fig:test2}
\end{minipage}
\caption{{\small (a) Comparison of the nominal state constraint sets of the robust MPC ($\bar{\calX}_R$), the benchmark SMPC ($\bar{\calX}_S$), and the adaptive SMPC for $10,100$ and 1000 noise samples ($\bar{\calX}_S^{10}$, $\bar{\calX}_S^{100}$, $\bar{\calX}_S^{1000}$). (b) Comparison of the terminal constraint sets of the robust MPC ($\bar{\calX}_{f,R}$), the benchmark SMPC ($\bar{\calX}_{f,S}$), and the adaptive SMPC for $10,100$ and 1000 noise samples ( $\bar{\calX}_{f,S}^{10}$, $\bar{\calX}_{f,S}^{100}$, $\bar{\calX}_{f,S}^{1000}$).}}
\end{figure}

To demonstrate the benefits of our adaptive SMPC scheme, we compare it with the robust MPC and the benchmark SMPC in a two-target tracking scenario for initial state $x_0=\vect{6}{4}^T$. 
More specifically, suppose we want to employ the DC-DC converter in a DC motor drive to track $x_{\textup{target}}^{(1)}=\vect{0}{0}^T$ until time step $t=100$, when the target changes to $x_{\textup{target}}^{(2)}=\vect{-0.1169}{3.9983}^T$. In order to track the two consecutive desired states, we sequentially apply each MPC algorithm twice: i) with target $x_{\textup{target}}^{(1)}$, for $t=0,\ldots,99$ (Phase 1), and ii) with target $x_{\textup{target}}^{(2)}$, for $t=100,101,\ldots$ (Phase 2). We will evaluate the performance of our adaptive SMPC scheme using the cumulative cost $C(t_0,t_0+T) = \sum_{t=t_0}^{t_0+T-1}\ell(x_t,u_t)$, where $t_0$ is the initial time, $T$ the running time, and $x_t$, $u_t$ the state and the input, respectively, of the system in closed-loop with a model predictive controller.

In simulations we observed that robust MPC is feasible in Phase 1 and infeasible in Phase 2, in contrast to our benchmark and adaptive SMPC schemes, which are feasible in both phases. This indicates the enlarged domain of attraction of problems \eqref{BSMPC} and \eqref{ASMPC} compared to the one of problem \eqref{RMPC} in Phase 2, by the beginning of which $100$ noise samples have been gathered online (see Remark 5). Figure 3 illustrates the state evolution of the system under the benchmark and the adaptive SMPC schemes for $30$ noise sequences $\{w_0,\cdots,w_{118}\}$, as well as part of the boundary of the hard state constraint set $\calX$. We observe that the closed-loop trajectories of the two schemes are more and more similar over time (in Phase 2, at the beginning of which $100$ noise samples are available, they are almost indistinguishable), and in each phase the system converges to a small neighborhood of the corresponding target, thus verifying our robust stability guarantees. Notice also that the hard constraints $x_{t,1}\leq9.6$ and $x_{t,1}\geq-9.6$ are violated in Phase 1 and 2, respectively.

\begin{figure}[tbh]
   \centering
   \includegraphics[width=0.5\linewidth]{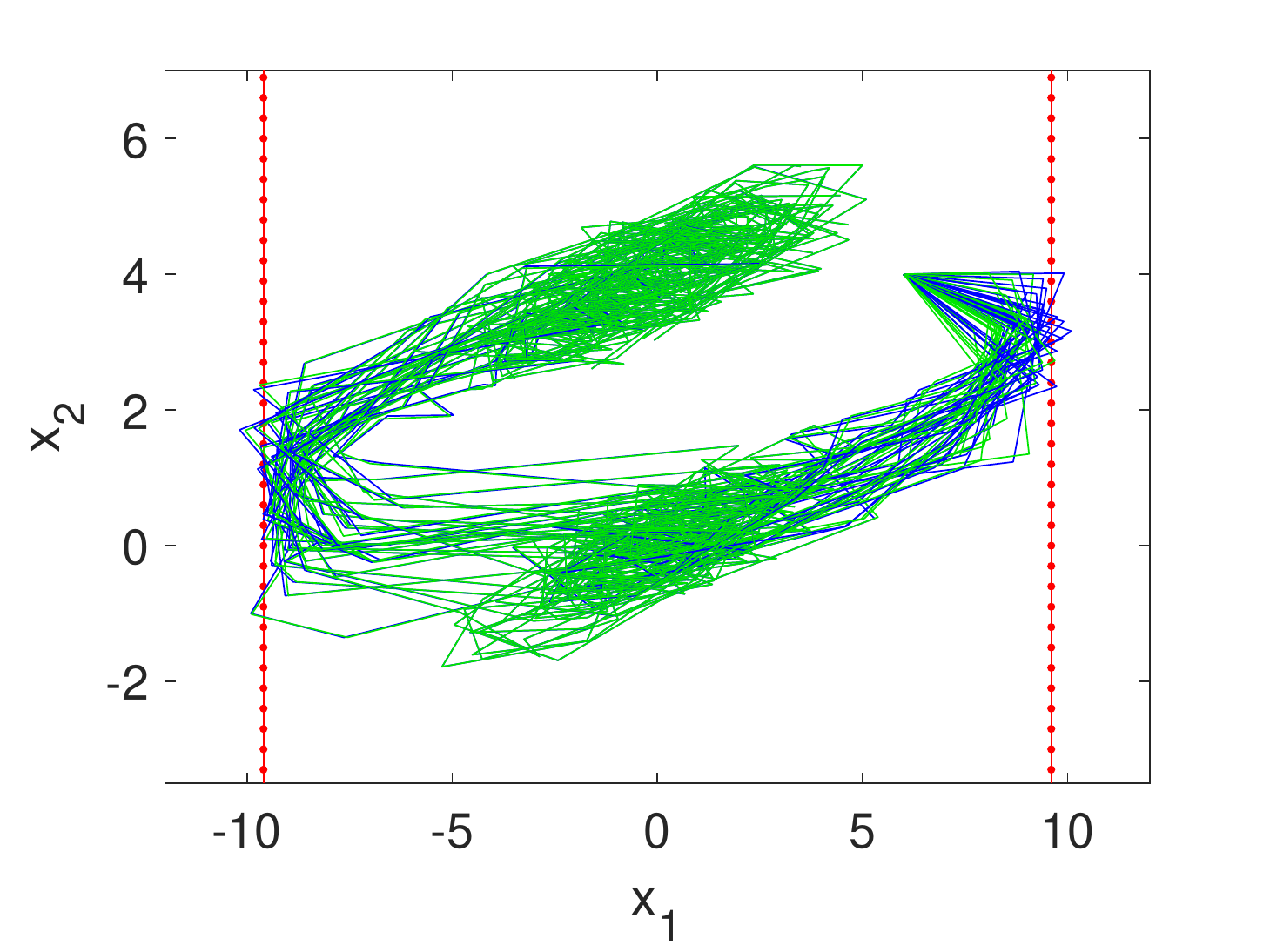}
   \caption{{\small State evolution of the DC-DC converter system under the benchmark SMPC (blue) and the adaptive SMPC (green) for initial state $x_0=\vect{6}{4}^T$ and $30$ noise sequences. The target state is $x_{\textup{target}}^{(1)}=\vect{0}{0}^T$ until $t=100$, when it is changed to $x_{\textup{target}}^{(2)}=\vect{-0.1169}{3.9983}^T$. The trajectories move clockwise from $(6,4)$ to $(0,0)$, and then to $(-0.1169,3.9983)$. The red lines belong to the boundary of the hard state constraint set $\calX$. Robust MPC is not depicted as it becomes infeasible when the target changes to $x_{\textup{target}}^{(2)}$.}}
\end{figure}

\begin{table}[tbh]
\centering
{\small
\begin{tabular}{|c|c|c|c|c|}
     \hline
     & \makecell{Violation of $x_{t,1}\leq9.6$ \\
    in Phase 1 ($\%$)} & \makecell{Violation of $x_{t,1}\geq-9.6$\\ in Phase 2 ($\%$)} &
     \makecell{$C(0,6)$} & \makecell{$C(100,106)$} \\
     \hline
     Benchmark SMPC & 6.24 & 5.60 & 515.73 & 788.19 \\
     \hline
     Adaptive SMPC & 0.32 & 3.92 & 518.62 & 788.50 \\
     \hline
\end{tabular}}
\caption{{\small Average constraint violation and cumulative cost over $10^3$ noise sequences.}}
\end{table}
Simulating the system with the benchmark and the adaptive SMPC for $10^3$ noise sequences, we computed the average constraint violations and cumulative costs shown in Table 1. Note that these values were calculated only for the first $7$ time steps of each phase, which was the time needed for the system to converge. Notice that in both phases the average constraint violation is larger and the average cumulative cost is smaller in the case of benchmark SMPC. This can be explained by the fact that in each phase the benchmark SMPC problem \eqref{BSMPC} has looser nominal state and terminal constraints than the adaptive SMPC problem \eqref{ASMPC}, and hence a larger feasible set, which allows the benchmark controller to violate the hard state constraints more often and achieve smaller cost than the adaptive one. However, in Phase 2, by the beginning of which $100$ noise samples have been collected online, the difference in the average constraint violation and cumulative cost of the benchmark and the adaptive SMPC is much smaller than in Phase 1. This was expected given our theoretical remark that the performance of the adaptive controller improves over time, as more noise samples are gathered and better estimates of the constraints \eqref{BSMPC_nominal_state_constraints} are obtained, gradually converging to the performance of the benchmark controller (see Remark 4).

\section{Conclusion and Future Work}
In this paper, we studied the problem of stochastic MPC (SMPC) for constrained linear systems under unknown noise distribution. We developed the first robustly stable adaptive SMPC scheme that guarantees time-uniform satisfaction of unknown chance state constraints with high probability, while learning their deterministic counterparts and thus improving control performance online.

Moving forward, our paper opens up several research directions. First, in order to reduce the conservatism introduced by the use of a tube-based control law with a precomputed stabilizing gain, we could try different parametrizations of the feedback control policy (e.g., a disturbance feedback parametrization with the gain as an optimization variable). Moreover, to mitigate possible conservatism resulting from the use of distributionally robust reformulated state constraints, we could try a new reformulation based on the quantile instead of the noise statistics, and provide time-uniform confidence intervals for the empirical quantile to design the new adaptive model predictive controller. Finally, our adaptive SMPC scheme could be used to analyze the regret of stochastic MPC under unknown noise distribution, which is an open problem.

\section*{Acknowledgement}
This work was supported by the AFOSR under grant FA9550-19-1-0265 (Assured Autonomy in Contested Environments).

\bibliography{arxiv_version}

\appendix
\section{Proofs}
\subsection{Proof of Lemma 1}
Using \eqref{error_system}, we can write:
\begin{align*}
    x_{t+k|t} &= \bar{x}_{t+k|t}+e_{t+k|t} \nonumber\\
              &= \bar{x}_{t+k|t}+A_{cl}e_{t+k-1|t}+w_{t+k-1|t},
\end{align*}
and thus replacing $x_t$ with $x_{t+k|t}$ in \eqref{chance_state_constraints}, we get:
\begin{align}\label{eq1}
    \Prob\Big(g_j^T\left(\bar{x}_{t+k|t}+A_{cl}e_{t+k-1|t}+w_{t+k-1|t}\right)&\leq h_j\Big)\geq1-\alpha_j.
\end{align}
Assuming that $x_t\in\bar{x}_{t|t}\oplus\calE$, we have $e_{t|t}\in\calE$ and by invariance of $\calE$ for system \eqref{error_system}, we deduce that $e_{t+k-1|t}\in\calE$, for all $k=1,\ldots,N$. Therefore, in order for \eqref{eq1} to hold, it suffices to have:
\begin{equation}\label{eq2}
    \Prob\left(g_j^Tw_{t+k-1|t}\leq h_j-g_j^T\bar{x}_{t+k|t}-\gamma_j\right)\geq1-\alpha_j.
\end{equation}
Employing Cantelli's inequality, we obtain:
\begin{align*}
    \Prob\left(g_j^Tw_{t+k-1|t}\geq h_j-g_j^T\bar{x}_{t+k|t}-\gamma_j\right)\leq \frac{g_j^T\Sigma_wg_j}{g_j^T\Sigma_wg_j+(h_j-g_j^T\bar{x}_{t+k|t}-\gamma_j)^2}.
\end{align*}
Consequently, a sufficient condition for \eqref{eq2} is:
\begin{align*}
    \alpha_j\geq \frac{g_j^T\Sigma_wg_j}{g_j^T\Sigma_wg_j+(h_j-g_j^T\bar{x}_{t+k|t}-\gamma_j)^2},
\end{align*}
which is equivalent to \eqref{pre_BSMPC_nominal_state_constraints}.$\qed$

\subsection{Proof of Theorem 1}
The proof of recursive feasibility and robust stability is omitted as it is similar to the one of the robust tube-based MPC \citep[Theorem 1]{Mayne2005}. 

From the analysis leading to \eqref{BSMPC_nominal_state_constraints} we know that the nominal state constraints \eqref{BSMPC_nominal_state_constraints} along with the initial constraint $x_t\in\bar{x}_{t|t}\oplus\calE$ of \eqref{BSMPC} are sufficient for \eqref{chance_state_constraints}. Combining this fact with recursive feasibility of \eqref{BSMPC}, we conclude that satisfaction of \eqref{chance_state_constraints} is guaranteed for all $t\geq0$, given the assumption that \eqref{BSMPC} is feasible at $t=0$.$\qed$

\subsection{Proof of Lemma 2}
Since the random vectors $w_i$ are zero-mean i.i.d. with covariance $\Sigma_w$, the variables $g_j^Tw_iw_i^Tg_j$ are i.i.d. with mean $g_j^T\Sigma_wg_j$. Moreover, given that $w_i$ satisfy $\norm{w_i}_\infty\leq a_w$, we have:
\begin{equation*}
    0\leq g_j^Tw_iw_i^Tg_j = (g_j^Tw_i)^2\leq\norm{g_j}_1^2a_w^2.
\end{equation*}
Hence, by definition of $\hat{\Sigma}_w^t$ and Hoeffding's inequality we obtain:
\begin{equation*}
    \Prob\left[-g_j^T\hat{\Sigma}_w^tg_j+g_j^T\Sigma_wg_j\geq s\right]\leq\exp\left(-\frac{2ts^2}{\norm{g_j}_1^4a_w^4}\right),
\end{equation*}
for all $s>0$. Therefore, setting $\delta'=t^2j^2\exp\left(-\frac{2ts^2}{\norm{g_j}_1^4a_w^4}\right)$, we conclude that with probability at least $1-\frac{\delta'}{t^2j^2}$:
\begin{equation}\label{eq3}
    g_j^T\Sigma_wg_j\leq g_j^T\hat{\Sigma}_w^tg_j+\norm{g_j}_1^2a_w^2\sqrt{\frac{1}{2t}\log\left(\frac{t^2j^2}{\delta'}\right)}.
\end{equation}
If $\calC_{tj}(\delta')$ denotes the set of matrices $\hat{\Sigma}_w^t$ that satisfy \eqref{eq3}, by De Morgan's law and Boole's inequality we can write:
\begin{align*}
   \Prob\left(\hat{\Sigma}_w^t\in\bigcap_{t,j} \calC_{tj}(\delta')\right)&= 1-\Prob\left(\hat{\Sigma}_w^t\in\bigcup_{t,j} \calC_{tj}^c(\delta')\right) \geq1-\sum_{j=1}^{n_x}\sum_{t=1}^\infty \Prob\left(\hat{\Sigma}_w^t\in \calC_{tj}^c(\delta')\right) \\
    &= 1-\sum_{j=1}^{n_x}\frac{\delta'}{j^2}\sum_{t=1}^\infty\frac{1}{t^2}\geq1-\frac{\pi^2\delta'}{6}\sum_{j=1}^\infty\frac{1}{j^2} = 1-\frac{\pi^4\delta'}{36}.
\end{align*}
Consequently, setting $\delta = \frac{\pi^4\delta'}{36}$, we deduce that with probability at least $1-\delta$:
\begin{equation*}
    g_j^T\Sigma_wg_j\leq g_j^T\hat{\Sigma}_w^tg_j+\norm{g_j}_1^2a_w^2\sqrt{\frac{1}{2t}\log\left(\frac{\pi^4t^2n_x^2}{36\delta}\right)},
\end{equation*}
uniformly over all $t\geq1$ and $j=1,\ldots,n_x$, which completes the proof.$\qed$

\subsection{Proof of Theorem 2}
Assume that \eqref{ASMPC} is feasible at time step $t=0$. We will prove that the solution of \eqref{ASMPC} ensures recursive feasibility, that is, we will show that under the above assumption, \eqref{ASMPC} is feasible at all time steps $t\geq0$. Let $\{\bar{u}_{0|0}^*,\ldots,\bar{u}_{N-1|0}^*\}$, $\{\bar{x}_{0|0},\ldots,\bar{x}_{N|0}\}$ denote the optimal nominal control and nominal state sequence, respectively. We apply $u_0$ given by \eqref{tube_controller} to \eqref{system} and the state evolves to:
\begin{align}\label{eq4}
    x_{1} &= Ax_0+BK(x_0-\bar{x}_{0|0})+B\bar{u}_{0|0}^*+w_0 \nonumber \\
          &= \bar{x}_{1|0}+A_{cl}e_{0|0}+w_0.
\end{align}
We will prove that at time step $t=1$ the sequences $\{\bar{u}_{1|0}^*,\ldots,\bar{u}_{N-1|0}^*,\bar{K}\bar{x}_{N|0}\}$ and $\{\bar{x}_{1|0},\ldots,\bar{x}_{N|0},$ $\bar{A}_{cl}\bar{x}_{N|0}\}$ compose a feasible solution of \eqref{ASMPC}.
From the initial condition of \eqref{ASMPC} we have $e_{0|0}\in\calE$. Hence, by invariance of $\calE$ for system \eqref{error_system}, we obtain $x_1\in\bar{x}_{1|0}\oplus\calE$, given \eqref{eq4}. It is also clear that $\bar{u}_{1|0}^*,\ldots,\bar{u}_{N-1|0}^*\in\calU\ominus K\calE$, $\bar{x}_{1|0},\ldots,\bar{x}_{N-1|0}\in\bar{\calX}_S^0\subseteq\bar{\calX}_S^1$, and $\bar{x}_{N|0}\in\bar{\calX}_{f,S}^0\subseteq\bar{\calX}_S^0\subseteq\bar{\calX}_S^1$. Moreover, since $\bar{x}_{N|0}\in\bar{\calX}_{f,S}^0$, by the definition of $\bar{\calX}_{f,S}^0$ we deduce that  $\bar{K}\bar{x}_{N|0}\in\calU\ominus K\calE$ and $\bar{A}_{cl}\bar{x}_{N|0}\in\bar{\calX}_{f,S}^0\subseteq\bar{\calX}_{f,S}^1$, which completes the proof of recursive feasibility (the proof is similar for the next time steps). 

By construction, the nominal state constraints \eqref{ASMPC_nominal_state_constraints} of problem \eqref{ASMPC} are sufficient for \eqref{BSMPC_nominal_state_constraints} uniformly over all $t\geq0$, with probability at least $1-\delta$. Moreover, from the analysis in Section 4 we know that the nominal state constraints \eqref{BSMPC_nominal_state_constraints} of \eqref{BSMPC} along with the initial constraint $x_t\in\bar{x}_{t|t}\oplus\calE$ are sufficient for \eqref{chance_state_constraints}. Combining these facts with recursive feasibility of \eqref{ASMPC}, we conclude that satisfaction of \eqref{chance_state_constraints} is guaranteed uniformly over all $t\geq0$, with probability at least $1-\delta$.

Following similar arguments to \citep[Theorem 1]{Mayne2005}, we will now show that applying the model predictive controller defined by \eqref{ASMPC} and \eqref{tube_controller} to system \eqref{system} in closed-loop implies asymptotic stability of the nominal system. Let $\calX_0$ denote the set of feasible initial states and $J_t^*(x_t)$ the optimal value function of \eqref{ASMPC}. Then we have:
\begin{equation*}
    J_0^*(x_0) = \sum_{k=0}^{N-1} \ell(\bar{x}_{k|0},\bar{u}_{k|0}^*)+\bar{x}_{N|0}^TP\bar{x}_{N|0},
\end{equation*}
which clearly satisfies $J_0^*(0)=0$ and $J_0^*(x_0)>0,\;\forall x_0\in\calX_0\setminus\{0\}$. Let $\tilde{J}_1(x_1)$ denote the value function obtained by applying the suboptimal nominal control sequence $\{\bar{u}_{1|0}^*,\ldots,\bar{u}_{N-1|0}^*,\bar{K}\bar{x}_{N|0}\}$ at time step $t=1$. By performing straightforward algebraic manipulations, we get:
\begin{align*}
    J_1^*(x_1) \leq J_0^*(x_0)-\ell(\bar{x}_{0|0},\bar{u}_{0|0}^*)+\bar{x}_{N|0}^T\Big(\bar{A}_{cl}^TP\bar{A}_{cl}-P+Q+\bar{K}^TR\bar{K}\Big)\bar{x}_{N|0}<J_0^*(x_0),
\end{align*}
for all $x_0\in\calX_0\setminus\{0\}$, where the final step follows from positive definiteness of $\ell(\cdot,\cdot)$ and \eqref{lyapunov_equation}. Consequently, $J_0^*(x_0)$ is a Lyapunov function, which implies that the nominal system is asymptotically stable. Therefore, system \eqref{system} in closed-loop with the model predictive controller determined by \eqref{ASMPC} and \eqref{tube_controller} is robustly stable in the sense that $\lim_{t\to\infty}\dist\left(x_t,\calE\right)=0$, where $\dist(\cdot,\cdot)$ is any fixed distance function.$\qed$

\section{Extension for Unknown Noise Mean}
For clarity of presentation, we assumed that the noise mean is known and equal to zero. Now we explain how this simplifying assumption can be dropped, thus extending our results to the case of unknown noise mean $\mu_w$. 

In the case of benchmark SMPC, where $\mu_w$ is assumed to be known, the nominal state constraints \eqref{BSMPC_nominal_state_constraints} are modified as follows:
\begin{align}\label{eq5}
    g_j^T\bar{x}_{t+k|t}\leq h_j-\gamma_j-\min\left\{\norm{g_j}_1a_w,g_j^T\mu_w+f(\alpha_j)\sqrt{g_j^T\Sigma_wg_j}\right\},
\end{align}
with the extra term $g_j^T\mu_w$ now introduced when applying Cantelli's inequality as in the proof of Lemma 1. 

Let $\hat{\mu}_w^t=\frac{1}{t}\sum_{i=0}^{t-1}w_i$, $\hat{\Sigma}_w^t=\frac{1}{t}\sum_{i=0}^{t-1}(w_i-\mu_w)(w_i-\mu_w)^T$ and $\tilde{\Sigma}_w^t=\frac{1}{t}\sum_{i=0}^{t-1}(w_i-\hat{\mu}_w^t)(w_i-\hat{\mu}_w^t)^T$ denote the noise sample mean, the noise sample covariance under known noise mean, and the noise sample covariance under unknown noise mean, respectively, at time step $t$. 
\begin{mylemma}
Fix a failure probability $\delta\in(0,1)$. With probability at least $1-\delta$, the sample mean $\hat{\mu}_w^t$ and the sample covariance $\tilde{\Sigma}_w^t$ satisfy:
\begin{equation*}
   g_j^T\mu_w\leq g_j^T\hat{\mu}_w^t+\tilde{r}_{tj}(\delta/2) 
\end{equation*}
and:
\begin{equation*}
    g_j^T\Sigma_wg_j\leq g_j^T\tilde{\Sigma}_w^tg_j+4r_{tj}(\delta/2)+\tilde{r}_{tj}^2(\delta/2),
\end{equation*}
respectively, where $r_{tj}(\cdot)$ is given by \eqref{ASMPC} and:
\begin{align*}
   \tilde{r}_{tj}(\delta) = \norm{g_j}_1a_w\sqrt{\frac{2}{t}\log\left(\frac{\pi^4t^2n_x^2}{18\delta}\right)},
\end{align*}
uniformly over all $t\geq1$ and $j=1,\ldots,n_x$.
\end{mylemma}
\begin{proof}
By performing straightforward algebraic manipulations, we can easily prove that:
\begin{equation}\label{eq6}
    \hat{\Sigma}_w^t=\tilde{\Sigma}_w^t+(\mu_w-\hat{\mu}_w^t)(\mu_w-\hat{\mu}_w^t)^T.
\end{equation} 
Since the random variables $g_j^Tw_i$ are independent and satisfy $\abs{g_j^Tw_i}\leq\norm{g_j}_1a_w$, by definition of $\hat{\mu}_w^t$ and Hoeffding's inequality we obtain:
\begin{equation*}
    \Prob\left[\abs{g_j^T\hat{\mu}_w^t-g_j^T\mu_w}\geq s\right]\leq 2\exp\left(-\frac{s^2t}{2\norm{g_j}_1^2a_w^2}\right),
\end{equation*}
for all $s>0$. Hence, setting $\delta'=2t^2j^2\exp\left(-\frac{s^2t}{2\norm{g_j}_1^2a_w^2}\right)$, we deduce that with probability at least $1-\frac{\delta'}{t^2j^2}$:
\begin{equation}\label{eq7}
    \abs{g_j^T\hat{\mu}_w^t-g_j^T\mu_w}\leq\norm{g_j}_1a_w\sqrt{\frac{2}{t}\log\left(\frac{2t^2j^2}{\delta'}\right)}.
\end{equation}
If $\calM_{tj}(\delta')$ denotes the set of vectors $\hat{\mu}_w^t$ that satisfy \eqref{eq7}, by De Morgan's law and Boole's inequality we can write:
\begin{align*}
    \Prob\left(\hat{\mu}_w^t\in\bigcap_{t,j} \calM_{tj}(\delta')\right)&= 1-\Prob\left(\hat{\mu}_w^t\in\bigcup_{t,j} \calM_{tj}^c(\delta')\right) \geq1-\sum_{j=1}^{n_x}\sum_{t=1}^\infty \Prob\left(\hat{\mu}_w^t\in \calM_{tj}^c(\delta')\right) \\
    &= 1-\sum_{j=1}^{n_x}\frac{\delta'}{j^2}\sum_{t=1}^\infty\frac{1}{t^2}\geq1-\frac{\pi^2\delta'}{6}\sum_{j=1}^\infty\frac{1}{j^2} = 1-\frac{\pi^4\delta'}{36}.
\end{align*}
Therefore, setting $\delta = \frac{\pi^4\delta'}{36}$, we conclude that with probability at least $1-\delta$:
\begin{equation}\label{eq8}
    \abs{g_j^T\hat{\mu}_w^t-g_j^T\mu_w}\leq\tilde{r}_{tj}(\delta),
\end{equation}
uniformly over all $t\geq1$ and $j=1,\ldots,n_x$. Similarly to the proof of Lemma 2, we can show that with probability at least $1-\delta$:
\begin{equation}\label{eq9}
    g_j^T\Sigma_wg_j\leq g_j^T\hat{\Sigma}_w^tg_j+4r_{tj}(\delta),
\end{equation}
uniformly over all $t\geq1$ and $j=1,\ldots,n_x$. Note that the extra coefficient $4$ in the last term of the right-hand side of (25) results from the fact that Hoeffding's inequality is now applied for the random variables $g_j^T(w_i-\mu_w)(w_i-\mu_w)^Tg_j$, which are upper-bounded by $4\norm{g_j}_1^2a_w^2$, instead of the random variables $g_j^Tw_iw_i^Tg_j$, which are upper-bounded by $\norm{g_j}_1^2a_w^2$. Consequently, by a union bound we can show that both \eqref{eq8} and \eqref{eq9} hold uniformly over all $t\geq1$ and $j=1,\ldots,n_x$, with probability at least $1-2\delta$. Employing this fact as well as \eqref{eq6}, the proof is easily completed.
\end{proof}
From Lemma 3 and \eqref{eq5} the proposed adaptive SMPC scheme can be extended to the case of unknown noise mean by simply replacing the nominal state constraints \eqref{ASMPC_nominal_state_constraints} with:
\begin{align*}
 &g_j^T\bar{x}_{t+k|t}\leq h_j-\gamma_j- \min\Bigg\{\norm{g_j}_1a_w,\min_{1\leq\tau\leq t}\Bigg[g_j^T\hat{\mu}_w^\tau+\\
 &\tilde{r}_{\tau j}(\delta/2)+f(\alpha_j)\sqrt{g_j^T\tilde{\Sigma}_w^\tau g_j+4r_{\tau j}(\delta/2)+\tilde{r}_{\tau j}^2(\delta/2)}\Bigg]\Bigg\},
\end{align*}
\hspace{-0.135cm}and all the observations made in Section 5 still apply.

\end{document}